\documentclass[11pt]{article}
\usepackage[T1]{fontenc}
\usepackage{babel}
\usepackage{verbatim}
\usepackage{bm}
\usepackage{amsmath,amstext,amsthm,amssymb}
\usepackage{float}
\usepackage{xspace}
\usepackage{graphicx}
\usepackage{mathtools}
\usepackage{subcaption}
\usepackage{wrapfig}
\usepackage{fullpage}

\usepackage[dvipsnames]{xcolor}
\definecolor{DarkGray}{rgb}{0.66, 0.66, 0.66}
\definecolor{DarkPowderBlue}{rgb}{0.0, 0.2, 0.6}
\definecolor{fluorescentyellow}{rgb}{0.8, 1.0, 0.0}

\usepackage{tikz-qtree,tikz-qtree-compat}
\usepackage{tikz}

\usepackage[unicode=true,
bookmarks=false,
breaklinks=false,pdfborder={0 0 1},backref=none,colorlinks=true,allcolors=blue]
{hyperref}

\usepackage{cleveref}
\usepackage{thmtools}
\usepackage{thm-restate}
\usepackage{bbm}

\usepackage[ruled,vlined,linesnumbered,algonl]{algorithm2e}
\SetEndCharOfAlgoLine{}
\SetKwComment{Comment}{\footnotesize$\triangleright$\ }{}

\SetCommentSty{mycommfont}

\Crefname{algocf}{Algorithm}{Algorithms}
\crefname{algocfline}{line}{lines}
\Crefname{invariant}{Invariant}{Invariants}
\Crefname{claim}{Claim}{Claims}
\Crefname{subclaim}{Subclaim}{Subclaims}

\usepackage{enumitem}
\usepackage{fullpage}
\usepackage[margin=1in]{geometry}
\usepackage{nicefrac}

\theoremstyle{plain}
\newtheorem{theorem}{Theorem}
\newtheorem{lemma}[theorem]{Lemma}

\newtheorem{cor}[theorem]{Corollary}
\newtheorem{remark}[theorem]{\protect\remarkname}
\theoremstyle{plain}
\theoremstyle{definition}
\newtheorem{defn}[theorem]{\protect\definitionname}
\theoremstyle{plain}
\theoremstyle{remark}

\ifx\proof\undefined
\newenvironment{proof}[1][\protect\proofname]{\par
	\normalfont\topsep6\p@\@plus6\p@\relax
	\trivlist
	\itemindent\parindent
	\item[\hskip\labelsep\scshape #1]\ignorespaces
}{%
	\endtrivlist\@endpefalse
}
\providecommand{\proofname}{Proof}
\fi

\usepackage{babel}
\providecommand{\definitionname}{Definition}

\providecommand{\claimname}{Claim}
\providecommand{\remarkname}{Remark}

\usepackage[textsize=tiny,textwidth=2cm,color=green!50!gray,obeyFinal]{todonotes}

\usepackage{soul}
\sethlcolor{fluorescentyellow}

\newcommand{\ie}{i.e.,\xspace}

\makeatother

\providecommand{\claimname}{Claim}
\providecommand{\definitionname}{Definition}

\newcommand{\sumL}{\sum\limits}

\newcommand{\eg}{e.g.,}

\newcommand{\EE}{\mathbb{E}}

\newcommand{\PP}{\mathbb{P}}

\newcommand{\RR}{\mathbb{R}}

\newcommand{\cB}{\mathcal{B}}

\newcommand{\cI}{\mathcal{I}}

\newcommand{\cM}{\mathcal{M}}

\newcommand{\cP}{\mathcal{P}}

\newcommand{\eps}{\varepsilon}

\newcommand{\poly}{\operatorname{poly}}
\newcommand{\opt}{\operatorname{OPT}}
\newcommand{\Var}{\operatorname{Var}}
\newcommand{\etal}[0]{\textit{et al.}}
\newcommand{\maxcut}{\operatorname{Max-Cut}}
\newcommand{\cmaxcut}{\operatorname{Constrained\;Max-Cut}}
\newcommand{\GW}{\operatorname{\sf GW}}

\definecolor{mygreen}{RGB}{10,150,110}
\definecolor{myred}{RGB}{150,10,20}
\providecommand{\email}[1]{\href{mailto:#1}{\nolinkurl{#1}\xspace}}

\title{Max-Cut with Multiple Cardinality Constraints}
\author{Yury Makarychev\thanks{Toyota Technological Institute at Chicago (TTIC): \email{yury@ttic.edu}. Supported by NSF Awards CCF-1955173, CCF-1934843, and ECCS-2216899.} \and Madhusudhan Reddy Pittu\thanks{Carnegie Mellon University: \email{mpittu@andrew.cmu.edu}.} \and Ali Vakilian\thanks{Toyota Technological Institute at Chicago (TTIC): \email{vakilian@ttic.edu}.}}
\date{}

\begin{document}

\maketitle

\begin{abstract}
    We study the classic Max-Cut problem under multiple cardinality constraints, which we refer to as the Constrained Max-Cut problem. Given a graph $G=(V, E)$, a partition of the vertices into $c$ disjoint parts $V_1, \ldots, V_c$, and cardinality parameters $k_1, \ldots, k_c$, the goal is to select a set $S \subseteq V$ such that $|S \cap V_i| = k_i$ for each $i \in [c]$, maximizing the total weight of edges crossing $S$ (i.e., edges with exactly one endpoint in $S$).

    By designing an approximate kernel for Constrained Max-Cut and building on the correlation rounding technique of Raghavendra and Tan (2012), we present a $(0.858 - \varepsilon)$-approximation algorithm for the problem when $c = O(1)$. The algorithm runs in time $O\big(\min\{k/\varepsilon, n\}^{\poly(c/\varepsilon)} + \poly(n)\big)$, where $k = \sum_{i \in [c]} k_i$ and $n=|V|$. This improves upon the $(\frac{1}{2} + \varepsilon_0)$-approximation of Feige and Langberg (2001) for $\maxcut_k$ (the special case when $c=1, k_1 = k$), and generalizes the $(0.858 - \varepsilon)$-approximation of Raghavendra and Tan (2012), which only applies when $\min\{k,n-k\}=\Omega(n)$ and does not handle multiple constraints.

    We also establish that, for general values of $c$, it is NP-hard to determine whether a feasible solution exists that cuts all edges. Finally, we present a $1/2$-approximation algorithm for Max-Cut under an arbitrary matroid constraint.
\end{abstract}

\newpage
\section{Introduction}
\label{sec:intro}
Given an undirected graph $G = (V, E)$ on $n$ vertices and a weight function $w: E \to \RR^+$, the \emph{Max-Cut} problem seeks a subset $S \subseteq V$ maximizing
$\delta_w(S) = \sum_{\substack{u \in S, v \in V \setminus S}} w(\{u,v\})$,
the total weight of edges crossing the cut $(S, V \setminus S)$. Without loss of generality, we assume the weights are scaled so that the total edge weight satisfies $\sum_{e \in E} w(e) = 1$.

Max-Cut is a fundamental problem in combinatorial optimization and approximation algorithms, with several landmark results, most notably the seminal SDP rounding algorithm by Goemans and Williamson~\cite{goemans1995improved}, which achieves an $\alpha_{\GW} \approx 0.878$ approximation. This approximation ratio is known to be optimal under the Unique Games Conjecture (UGC)~\cite{khot2007optimal}.

In this work, we study a variant called \emph{Constrained Max-Cut}, where additional partition constraints are imposed on the solution.

\begin{defn}[$\cmaxcut$]\label{def:constrained-max-cut}
Given a graph $G = (V, E)$, a weight function $w: E \to \RR^+$, and a set of $c$ partition constraints $\{(V_i, k_i)\}_{i\in [c]}$ where $V = \biguplus_{i\in [c]} V_i$ and $k_i \leq |V_i|/2$ for all $i$, the \emph{Constrained Max-Cut} problem asks to find a subset $S \subseteq V$ such that $|S \cap V_i| = k_i$ for all $i \in [c]$, maximizing $\delta_w(S)$.
\end{defn}

Several well-studied problems are special cases of $\cmaxcut$. The  \emph{Max-Bisection} problem corresponds to $c=1$ and $k=n/2$, and admits approximation factors close to $\alpha_{\GW}$---specifically, $0.8776$~\cite{austrin2016better} (see also~\cite{frieze1997improved,Ye01,feige2001note,HZ02,feige2006rpr2,GMRSZ11,raghavendra2012approximating}).
More generally, when there is a single cardinality constraint $|S|=k$ (\ie $c=1$), the problem is known as $\maxcut_k$~\cite{feige2001approximation}. It is also referred to as $(k, n-k)$-Max-Cut in parameterized complexity, see \eg~\cite{cai2008parameterized,saurabh2018k}.

\begin{defn}[$\maxcut_k$]
Given an undirected graph $G=(V,E)$, a weight function $w: E \to \RR^+$, and an integer $k$\footnote{Assume that $k\leq n/2$ without loss of generality.}, the $\maxcut_k$ problem seeks a subset $S \subseteq V$ of cardinality exactly $k$ that maximizes $\delta_w(S)$.
\end{defn}

For $k=\Omega(n)$, the global correlation rounding technique of Raghavendra and Tan~\cite{raghavendra2012approximating} (building on~\cite{barak2011rounding}) achieves an $\alpha_{\mathrm{cc}} \approx 0.858$ approximation. Austrin and Stankovic~\cite{austrin2019global} later showed that this approximation is essentially tight for $k < 0.365n$.
However, when $k=o(n)$, existing results are weaker. Feige and Langberg~\cite{feige2001approximation} gave a $(0.5 + \eps_0)$-approximation for all $k$, where $\eps_0$ is a small universal constant ($\eps_0 < 0.09$). Moreover, the pipage rounding technique of Ageev and Sviridenko~\cite{AS04} guarantees a $0.5$-approximation for all $k$.

\medskip
The special case of $\cmaxcut$ with $c=1$ and $k = o(n)$ has applications in {\em pricing in social networks}~\cite{candogan2012optimal}, also referred to as {\em influence-and-exploit}~\cite{fotakis2014efficiency}. In this context, consumers's valuation depends directly on the usage of their neighbors in the network. Consequently, the seller's optimal pricing strategy may involve offering discounts to certain influencers who hold central positions in the underlying network. Candogan~\etal~\cite{candogan2012optimal} considered a setting with two price types: {\em full} and {\em discounted}. Specifically, the objective is to maximize the total network influence, subject to the constraint of offering $k$ discounted prices to a small target set of buyers, where $k \ll n$. Candogan~\etal~showed that the obtained problem can be reduced to $\maxcut_k$, where $k = o(n)$. 

Moreover, in certain settings where diversity among influencers is desirable, it is natural to require that the selected influencers {\em fairly represent} different groups. This requirement can be modeled as a $\cmaxcut$ problem with multiple capacity constraints. In most relevant cases, the number of groups is a small integer $c = O(1)$. 
For a comprehensive survey on {\em fair representation} for learning on graphs, see to~\cite{laclau2022survey}.

\medskip
In this paper, we also introduce and study a more general version of the problem: finding a maximum cut subject to a matroid constraint.
\begin{defn}[Matroid Max-Cut]\label{def:maxcut-matroid}
    Given an undirected graph $G=(V,E)$, a weight function $w:E\rightarrow \mathbb{R}^+$, and a matroid $\cM=(V,\cI)$, the {\em matroid} Max-Cut problem is to maximize $\delta_w(S)$ subject to $S \in \cB$, where $\cB$ is the collection of bases of $\cM$. 
\end{defn}
Note that $\cmaxcut$ and $\maxcut_k$ are special cases of the matroid Max-Cut problem; we get these problems when the matroid is a partition matroid or a uniform matroid, respectively. This problem has not been explicitly studied previously. However, an algorithm by Lee, Mirrokni, Nagarajan, and Sviridenko~\cite{LMNS09} gives a $(\frac{1}{3}-\varepsilon)$-approximation to the more general problem of \emph{symmetric submodular function maximization} subject to matroid base constraint.

\subsection{Our Results and Techniques}
\label{subsec:our_results}
Our algorithm for $\maxcut_k$ builds on the global correlation rounding technique introduced by Raghavendra and Tan~\cite{raghavendra2012approximating}, which achieves an $\alpha_{cc} \approx 0.858$-approximation in the regime where $k = \Omega(n)$. We extend this approach by developing an approximate kernel and applying it in conjunction with correlation rounding, allowing us to handle the challenging case where $k = o(n)$. 

\smallskip
This yields an approximation guarantee of $(0.858 - \varepsilon)$ for $\maxcut_k$ across all values of $k$, improving upon the $(0.5 + \varepsilon_0)$-approximation of~\cite{feige2001approximation} in the sparse regime. Formally:

\begin{theorem}\label{thm:Max-Cut-k}
    For every $\varepsilon > 0$, there is an algorithm that runs in $O(\min\{k/\varepsilon, n\}^{\poly(1/\varepsilon)} + \poly(n))$ time and obtains an $(\alpha_{cc} - \varepsilon)$-approximation to the optimal $\maxcut_k$ solution, where $\alpha_{cc} \approx 0.858$ is the approximation factor of the Raghavendra--Tan algorithm.
\end{theorem}

The regime where $k = o(n)$ is particularly challenging, as the correlation rounding approach of Raghavendra and Tan~\cite{raghavendra2012approximating} does not extend to this setting. Our algorithm closes this gap by improving upon the $(0.5 + \varepsilon_0)$-approximation of~\cite{feige2001approximation} in the sparse regime. For completeness, we provide a brief overview of the approach of Raghavendra and Tan in \Cref{subsec:prelim}, and highlight the key reasons why it breaks down when $k = o(n)$.

\smallskip
We next address the more general setting of multiple constraints, focusing on the case $c = O(1)$.

\begin{theorem}\label{thm:Max-Cut-constant-constraints}
    For every $\varepsilon > 0$, there is an algorithm that runs in $O\big(\min\{k/\varepsilon, n\}^{\poly(c/\varepsilon)} + \poly(n)\big)$ time, where $k = \sumL_{i\in[c]} k_i$, and obtains an $(\alpha_{cc} - \varepsilon)$-approximation to the optimal solution of \textnormal{Constrained Max-Cut}. In particular, when $c = O(1)$ and $\varepsilon$ is fixed, the running time is polynomial in $n$.
\end{theorem}
More broadly, we study Max-Cut under an arbitrary matroid constraint $\mathcal{M} = (V, \mathcal{I})$, generalizing $\cmaxcut$ with an arbitrary number of partition constraints, especially when $c = \omega(1)$.

\begin{theorem}\label{thm:Max-Cut-matroid}
    There exists a $0.5$-approximation algorithm for matroid Max-Cut.
\end{theorem}

The only prior result for this setting is by Lee~\etal~\cite{LMNS09}, who provided a $(\frac{1}{3} - \varepsilon)$-approximation for the more general problem of symmetric submodular function maximization subject to a matroid base constraint.

Finally, we show that for general $\cmaxcut$ with an arbitrary number of constraints, it is NP-hard to decide whether there exists a feasible solution cutting all edges. Formally:

\begin{theorem}\label{thm:Max-Cut-hardness}
    Given a graph $G = (V, E)$, a partition of vertices into $V_1, \dots, V_c$, and budget parameters $k_1, \dots, k_c$, it is NP-hard to decide whether there exists a feasible solution $S$ such that $\delta(S) = |E|$.
\end{theorem}

We note that for the standard Max-Cut problem, this decision variant can be solved in polynomial time using bipartite testing.

\paragraph{Our Techniques.}
A key technical contribution of our work is the construction of an \emph{approximate kernel} for $\cmaxcut$. Specifically, for any cardinality constraint $k$, we sort the vertices by their (weighted) degrees as $v_1, \dots, v_n$, and define $\widetilde{V}$ as the top $O(k/\varepsilon)$ vertices. While the graph $G$ remains unchanged, we restrict our attention to solutions contained entirely within $\widetilde{V}$. Then, an optimal solution $\widetilde{S}$ to $\maxcut_k$ over $\widetilde{V}$ achieves a cut value that is at least a $(1 - \varepsilon)$ fraction of the true optimum. In other words,
\begin{align}
    \max_{\widetilde{S}\subseteq \widetilde{V},\,|\widetilde{S}|=k}\delta_{w}(\widetilde{S})\geq (1-\varepsilon)\cdot \max_{S\subseteq V,\,|S|=k}\delta_{w}(S).
\end{align}
See~\Cref{thm:maxcut_kernel_red_single} in~\Cref{sec:kernel} for the formal statement. 

\smallskip

This reduction is particularly useful because it allows us to focus on problem instances where $k = \Omega(n)$. Conceptually, we can contract the vertices in $V \setminus \widetilde{V}$ into a single super vertex $s$, and then restrict the solution to exclude $s$. This transforms the sparse regime into one where the effective solution size is a constant fraction of the (reduced) vertex set, enabling the use of correlation rounding techniques that require $k = \Omega(n)$.

In contrast, prior work by~\cite{saurabh2018k} uses kernelization to design fixed-parameter algorithms for $\maxcut_k$, but their parameter is the value of the optimal solution itself, and they aim for an \emph{exact} kernel. As a result, their kernel size is polynomial in $k$, which is insufficient for our purposes. Moreover, their kernelization is sequential and adaptive, while ours is non-adaptive. Our approximate kernel also extends to $\cmaxcut$ with multiple constraints ($c>1$), as formalized in~\Cref{thm:maxcut_kernel_red_double}.

\smallskip
Once we reduce to an instance with $k = \Omega(n)$, we apply the Raghavendra--Tan algorithm~\cite{raghavendra2012approximating} to obtain a subset of vertices of size $k' \in k \cdot [1 - \varepsilon, 1 + \varepsilon]$, achieving a cut value that is at least an $\alpha_{cc}$-fraction of the optimum. We then perform a \emph{random correction} step: adjusting the solution by randomly adding or removing at most $\varepsilon k$ vertices to exactly match the required size $k$, incurring only a negligible loss in cut value.

\smallskip
When $c > 1$, however, the rounding procedure of~\cite{raghavendra2012approximating} does not directly apply. To handle the setting with multiple constraints, we introduce the notion of \emph{$\alpha$-block independence} for SDP solutions, which generalizes the standard notion of $\alpha$-independence. Informally, an SDP solution is $\alpha$-block independent if, within each partition $V_i$, the average correlation between pairs of vertices is at most $\alpha$.

We first show how to efficiently construct a block-independent solution. Then, by applying the rounding algorithm of~\cite{raghavendra2012approximating}, we obtain a subset $S$ that approximately satisfies each group constraint: for each $i \in [c]$, the size $k_i' = |S \cap V_i|$ lies in the range $[(1 - \varepsilon)k_i, (1 + \varepsilon)k_i]$. Finally, we apply a random correction step within each group to enforce exact feasibility, while ensuring that the cut value degrades by only a negligible amount.

\smallskip
For matroid Max-Cut, we combine techniques from~\cite{AS04} and~\cite{CCPV11} to design a linear programming relaxation with integrality gap at most $0.5$, which can be solved efficiently. Applying pipage rounding to this relaxation yields a deterministic $0.5$-approximation algorithm.

\subsection{Preliminaries}
\label{subsec:prelim}
Our results heavily rely on the global correlation rounding technique developed in~\cite{raghavendra2012approximating}. For completeness, we include the relevant definitions and theorems in this section. A quick summary of the Lasserre hierarchy is provided in \Cref{sec:Lasserre-basics}.

Naive approaches based on variants of hyperplane rounding applied to a two-round Lasserre SDP relaxation for the $\maxcut_k$ problem can produce subsets $S$ of expected size $k$ that achieve good approximation guarantees. However, these approaches offer no control over the concentration of $|S|$ around $k$, due to potentially high correlations between the values assigned to vertices by the SDP solution. 

\paragraph{Notation.} We use $\mu = \{\mu_S\}_{|S| \leq \ell}$ to denote a level-$\ell$ Lasserre pseudo-distribution, where $\mu_S : \{-1,1\}^S \rightarrow [0,1]$ is a distribution over partial assignments to the subset $S \subseteq V$. Let $X_S$ denote the random variables jointly distributed according to $\mu_S$, and $X_i$ the marginal variable for $i \in V$ under $\mu_{\{i\}}$. We write $\PP_{\mu_S}[X_S \in A]$ to denote the pseudo-probability that the assignment to $S$ lies in the event $A \subseteq \{-1,1\}^S$. In particular, conditional pseudo-probabilities are expressed as
$\PP_{\mu_{S \cup \{i\}}}[X_i = 1 \mid X_S = \alpha]$,
which denotes the pseudo-probability that $X_i = 1$ given that $X_S = \alpha$ for some $\alpha \in \{-1,1\}^S$. 

\paragraph{SDP Relaxation.} To leverage the correlations between vertices, \cite{raghavendra2012approximating} employ an $\ell$-round Lasserre SDP for $\maxcut_k$ with a sufficiently large constant $\ell$, formally described in~\Cref{SDP:maxcut_k}. 

\begin{align}
 \label{SDP:maxcut_k}   
\max\quad & \sumL_{\{i,j\}\in E}w_{i,j}\PP_{\mu_{\{i,j\}}}[X_{\{i,j\}} \in \{(-1,1), (1,-1)\}]\\
\text{s.t.}\quad & \sumL_{i\in V}\PP_{\mu_{S\cup \{i\}}}(X_i=1 \mid X_S=\alpha)=k &&\forall S \subseteq V,\; |S|\leq \ell-1,\; \alpha \in \{0,1\}^S \nonumber\\ 
& \mu \textnormal{ is a level-}\ell \textnormal{ pseudo-distribution}. \nonumber
\end{align}

\paragraph{Measuring Correlations.} One method to assess the correlation between two random variables $X_i$ and $X_j$ is through mutual information, defined as $I_{\mu_{\{i,j\}}}(X_i; X_j) = H(X_i) - H(X_i \mid X_j)$, where $X_i$ and $X_j$ are sampled according to the local distribution $\mu_{\{i,j\}}$. An SDP solution is $\alpha$-\textit{independent} if the average mutual information between uniformly random vertex pairs is at most $\alpha$, i.e., $\mathbb{E}_{i,j \in V}[I(X_i; X_j)] \leq \alpha$.

\begin{defn}[$\alpha$-independence~\cite{raghavendra2012approximating}]
An SDP solution to an $\ell$-round Lasserre SDP is \emph{$\alpha$-independent} if $\EE_{i,j \in V}[I_{\mu_{\{i,j\}}}(X_i;X_j)] \leq \alpha$, where $\mu_{\{i,j\}}$ is the local distribution over $\{i,j\}$. More generally, if $W$ is a distribution over $V$, then the solution is $\alpha$-independent w.r.t. $W$ if $\EE_{i,j \sim W}[I_{\mu_{\{i,j\}}}(X_i;X_j)] \leq \alpha$. When unspecified, $W$ is assumed to be the uniform distribution over $V$.
\end{defn}

For many standard rounding schemes, such as halfspace rounding, the variance in the balance of the resulting cut is directly linked to the average correlation among random vertex pairs. Specifically, if the rounding scheme is applied to an $\alpha$-independent solution, the variance in the cut's balance is bounded by a polynomial function of $\alpha$.

\paragraph{Obtaining Uncorrelated SDP Solutions.} 
If all vertices in a $t$-round Lasserre SDP solution are highly correlated, conditioning on the value of one vertex reduces the entropy of the rest. Formally, if the solution is not $\alpha$-independent (i.e., $\EE_{i,j \in V}[I(X_i;X_j)] > \alpha$), then conditioning on a randomly chosen vertex $i$ and its value $b$ decreases the average entropy of the remaining variables by at least $\alpha$. Repeating this process $1/\alpha$ times suffices to obtain an $\alpha$-independent solution. Thus, starting from a $t$-round Lasserre SDP solution, this process results in a $(t - \ell)$-round $\alpha$-independent solution for some $\ell = O(1/\alpha)$.

\paragraph{Rounding Uncorrelated SDP Solutions.}
Given an $\alpha$-independent SDP solution, many natural rounding schemes ensure that the balance of the output cut is concentrated around its expectation. Hence, it suffices to construct rounding schemes that preserve the expected balance. Raghavendra and Tan~\cite{raghavendra2012approximating} present a simple rounding procedure that preserves the individual bias of each vertex, thereby ensuring the global balance property.

An elegant probabilistic argument from~\cite{raghavendra2012approximating} shows how to convert an $(\ell+4/\alpha^2+1)$-round Lasserre SDP solution into an $\alpha$-independent $\ell$-round solution, while losing only an additive factor of $\alpha$ in the objective value (assuming the optimum is normalized to at most $1$).

\begin{lemma}[\cite{raghavendra2012approximating}]
\label{lem:sampling_correlation_reduction}
There exists $t \leq k$ such that 
$\EE_{i_1,\dots,i_t \sim W}\EE_{i,j\sim W}\left[I(X_i;X_j \mid X_{i_1},\dots,X_{i_t})\right] \leq \frac{1}{k-1}$.
\end{lemma}

\Cref{lem:sampling_correlation_reduction} implies the existence of a $t \leq 1/\alpha + 1$ such that conditioning on the joint assignment to $t$ randomly sampled vertices reduces the average mutual information between other pairs to at most $\alpha$.

\begin{theorem}[\cite{raghavendra2012approximating}]
\label{thm:valuepreserving_independent_single}
For every $\alpha > 0$ and integer $\ell$, there exists an algorithm running in time $O(n^{\poly(1/\alpha) + \ell})$ that finds an $\alpha$-independent solution to the $\ell$-round Lasserre SDP, with objective value at least $\opt - \alpha$, where $\opt$ is the optimum SDP value.
\end{theorem}

\Cref{thm:valuepreserving_independent_single} implies that there exists $t = O(1/\alpha^2)$ such that conditioning on $t$ vertices yields an $\alpha$-independent solution with probability at least $\alpha/2$. Since the sampling procedure preserves the marginal biases of the vertices, the SDP objective remains close to optimal in expectation. By Markov’s inequality, the value of the conditioned solution is at least $\opt - \alpha$ with probability at least $1/(1+\alpha)$. Thus, there exists a small subset of vertices such that conditioning on them yields an $\alpha$-independent solution with near-optimal value.

Algorithm 5.3 of \cite{raghavendra2012approximating} is a rounding scheme that preserves the bias (according to the SDP solution) of every vertex while also approximately preserving the pairwise correlations up to polynomial factors. Using numerical techniques, they show that the probability of an edge being cut is at least $\alpha_{cc} \approx 0.858$ times its contribution to the SDP objective, implying that the total cut value is at least $\alpha_{cc}$ times the SDP value.

\paragraph{Controlling Cut Balance.} 

\begin{theorem}[\cite{raghavendra2012approximating}]
\label{thm:balance_variance}
Given an $\alpha$-independent solution to two rounds of the Lasserre SDP, let $\{y_i\}_{i \in V}$ denote the rounded output from Algorithm 5.3. Let $S = \EE_{i \sim W}[y_i]$ be the expected balance of the cut. Then,
$\Var(S) \leq O(\alpha^{1/12})$.
\end{theorem}

By applying Chebyshev’s inequality to \Cref{thm:balance_variance}, the number of vertices in the cut lies in the range $k \pm n \cdot O(\alpha^{1/24})$ with high probability. When $k/n = \Omega(1)$, we can choose $\alpha = \Omega(1)$ small enough so that the relative deviation is within $\varepsilon k$. A post-processing step can then adjust the set size to exactly $k$ (e.g., by swapping a small number of vertices), which incurs only an $O(\varepsilon)$ fractional loss in cut value. However, when $k = o(n)$, the additive error term $n \cdot O(\alpha^{1/24})$ may significantly exceed $k$, making it difficult to ensure cardinality feasibility without substantially affecting the objective.

\paragraph{Notation.} For any subset $S \subseteq V$ and vertex $i \in V$, we write $S + i := S \cup \{i\}$ and $S - i := S \setminus \{i\}$. Let $G = (V, E)$ be a graph with non-negative edge weights $w : E \rightarrow \mathbb{R}_{\geq 0}$. For a parameter $r \in \mathbb{N}$, let $H_r \subseteq V$ denote the set of the $r$ highest-degree vertices in $G$ under weight function $w$. If the vertex set $V$ is partitioned into $c$ disjoint groups, $V = \biguplus_{i \in [c]} V_i$, then $H_r^{(i)} \subseteq V_i$ denotes the $r$ highest-degree vertices in part $V_i$. When the weight function is clear from context, we abbreviate the weighted degree of a vertex $v$ as $\delta(v)$ instead of $\delta_w(v)$.

\section{Approximate Kernels for Max-Cut with Cardinality Constraints}
\label{sec:kernel}

In \Cref{subsec:kernel_max_cut_k}, we show that for any instance $(G, k)$ of $\maxcut_k$, one can reduce the graph to a (conditioned) instance $(\widetilde{G}, k)$ with $|\widetilde{V}| = O(k/\varepsilon)$ vertices. In \Cref{subsec:kernel_constrained_max_cut}, we generalize this construction to the setting with multiple partitions. Specifically, for any instance $(G = (V = \biguplus_{i \in [c]} V_i, E), (k_1, \dots, k_c))$ of $\cmaxcut$, we construct a conditioned instance $(\widetilde{G} = (\widetilde{V} = \biguplus_{i \in [c]} \widetilde{V}_i, \widetilde{E}), (k_1, \dots, k_c))$ such that for every $i \in [c]$, we have $|\widetilde{V}_i| = O(k_i / \varepsilon)$. 

\smallskip
We use $\opt$ to denote the optimal cut value of a given instance. For the single-group problem $\maxcut_k$ on a graph $G = (V, E)$, we define
$\opt := \max_{S \subseteq V,\ |S| = k} \delta(S)$.
For the multi-group case $\cmaxcut$ with partition $V = \biguplus_{i \in [c]} V_i$ and size constraints $k_1, \dots, k_c$, the optimal value is
$\opt := \max_{\substack{S = \biguplus_{i \in [c]} S_i \\ S_i \subseteq V_i,\ |S_i| = k_i}} \delta(S)$.

\subsection{Approximate Kernel for $\maxcut_k$}
\label{subsec:kernel_max_cut_k}

\paragraph{Kernel Procedure for $\maxcut_k$.}
We now describe the approximate kernel construction for the single-group $\maxcut_k$ problem.

\medskip
\noindent
\textbf{Input:} Graph $G = (V, E)$, cardinality parameter $k$, and approximation parameter $0 < \varepsilon\leq 1/2$.\\
\textbf{Output:} Reduced graph $\widetilde{G} = (\widetilde{V}, \widetilde{E})$.

\begin{enumerate}[leftmargin=*]
    \item If $k / \varepsilon \ge n$, return $G$.
    \item Otherwise, sort the vertices of $G$ in decreasing order of weighted degree, and retain only the top $k / \varepsilon$ vertices. Merge the rest of vertices into a super vertex $s$, and return the resulting graph $\widetilde{G}$.
\end{enumerate}

Note that the super vertex $s$ appears in the output graph $\widetilde{G}$ only if $k/\varepsilon + 1 \le n$.
\begin{theorem}
\label{thm:maxcut_kernel_red_single}
For any $\maxcut_k$ instance $(G, k)$, let $(\widetilde{G}, k)$ be the reduced instance returned by the \hyperref[subsec:kernel_max_cut_k]{$\maxcut_k$ kernel} procedure above. Then the optimal cut value of $\maxcut_k$ on $\widetilde{G}$, conditioned on not selecting the super vertex $s$, satisfies
\[
\widetilde{\opt} := \max_{S \subseteq \widetilde{V} \setminus \{s\},\, |S| = k} \delta(S) \ge (1 - 4\varepsilon) \cdot \opt.
\]
\end{theorem}

\begin{proof}
Let $h := k / \varepsilon$, and recall from the notation in the preliminaries that $H_h$ is the set of the $h$ highest-degree vertices in $G$. Let $S^*$ be an optimal solution of size $k$ in $G$ with cut value $\delta(S^*) = \opt$. We will construct a set $S_T \subseteq H_h$ of size $k$ with value close to $\opt$.

We iteratively transform $S^*$ into a set within $H_h$ by applying \Cref{lem:local_exchange} up to $k$ times. At each step $t$, we replace a vertex $j_t \in S_{t-1} \setminus H_h$ with a vertex $i_t \in H_h \setminus S_{t-1}$ such that:
\[
\delta(S_t) \ge \left(1 - \frac{2}{h - k}\right) \cdot \delta(S_{t-1}).
\]
Since each step increases $|S_t \cap H_h|$ by one, the process terminates in at most $T \le k$ steps. Therefore,
\[
\delta(S_T) \ge \opt \cdot \left(1 - \frac{2}{h - k}\right)^T \ge \opt \cdot \left(1 - \frac{2T}{h - k}\right) \ge \opt \cdot \left(1 - \frac{2k}{h - k}\right).
\]
Substituting $h = k / \varepsilon$, we get
\[
\delta(S_T) \ge \opt \cdot \left(1 - \frac{2\varepsilon}{1 - \varepsilon}\right) \ge \opt \cdot (1 - 4\varepsilon),
\]
 where the final bound assumes $\varepsilon \le 1/2$.
\end{proof}

\begin{lemma}
\label{lem:local_exchange}
Let $S \subseteq V$ be a subset of size $|S|=k \le h$ such that $S \setminus H_h \ne \emptyset$. Then there exist vertices $i \in H_h \setminus S$ and $j \in S \setminus H_h$ such that:
\[
\delta\big((S - j) + i\big) \ge \left(1 - \frac{2}{h - k}\right) \cdot \delta(S).
\]
\end{lemma}

\begin{proof}
Since $S \setminus H_h \ne \emptyset$ and $|S| \le h$, we have $H_h \setminus S \ne \emptyset$. Let $i \in H_h \setminus S$ be the vertex minimizing $\delta(S, \{i\})$, the total weight of edges between $i$ and $S$. Let $j$ be any vertex in $S \setminus H_h$.

We use the submodularity of the cut function:
\[
\delta((S - j) + i) + \delta(S) \ge \delta(S + i) + \delta(S - j).
\]
Rearranging:
\begin{align*}
\delta((S - j) + i) - \delta(S) 
&\ge \left[\delta(S + i) - \delta(S)\right] + \left[\delta(S - j) - \delta(S)\right] \\
&= \big(\delta(\{i\}) - 2\delta(S, \{i\})\big) - \big(\delta(\{j\}) - 2\delta(S - j, \{j\})\big) \\
&= \big(\delta(\{i\}) - \delta(\{j\})\big) + 2\delta(S - j, \{j\}) - 2\delta(S, \{i\}) \\
&\ge -2\delta(S, \{i\}).
\end{align*}

Now we bound $\delta(S, \{i\})$. Since $i$ minimizes $\delta(S, \cdot)$ among $H_h \setminus S$, we have:
\[
\delta(S) = \sum_{v \in V \setminus S} \delta(S, \{v\}) \ge \sum_{v \in H_h \setminus S} \delta(S, \{v\}) \ge |H_h \setminus S| \cdot \delta(S, \{i\}),
\]
which implies:
\[
\delta(S, \{i\}) \le \frac{\delta(S)}{h - k}.
\]
Putting everything together:
\[
\delta((S - j) + i) \ge \delta(S) - \frac{2}{h - k} \cdot \delta(S) = \left(1 - \frac{2}{h - k}\right) \cdot \delta(S),
\]
completing the proof.
\end{proof}

\subsection{Approximate Kernel for Constrained Max-Cut}
\label{subsec:kernel_constrained_max_cut}

\paragraph{Kernel Procedure for Constrained Max-Cut.}
\label{sec:kernal_procedure_double}
We now describe the kernelization procedure for the $\cmaxcut$ problem with multiple vertex groups.

\medskip
\noindent
\textbf{Input:} Graph $G=(V =\biguplus_{i \in [c]} V_i, E)$, cardinality constraints $k_1, \dots, k_c$, and approximation parameter $0 < \varepsilon\leq 1/2$. \\
\textbf{Output:} Reduced graph $\widetilde{G} = (\widetilde{V} = \biguplus_{i \in [c]} \widetilde{V}_i, \widetilde{E})$.

\begin{enumerate}
    \item For each $i \in [c]$, if $k_i / \varepsilon + 1 \le n_i := |V_i|$, retain the top $k_i / \varepsilon$ vertices in $V_i$ by weighted degree and merge the remaining vertices into a super vertex $s_i$.
    \item Return the resulting graph $\widetilde{G}$.
\end{enumerate}

Note that a super vertex $s_i$ appears in the output graph $\widetilde{G}$ only if $k_i / \varepsilon + 1 \le n_i$. Let $S_{\text{super}} := \{s_i \mid s_i \text{ exists in } \widetilde{G}\}$ denote the set of all super vertices.

\begin{theorem}
\label{thm:maxcut_kernel_red_double}
For any $\cmaxcut$ instance $(G, k_1, \dots, k_c)$, let $(\widetilde{G}, k_1, \dots, k_c)$ be the reduced instance returned by the \hyperref[sec:kernal_procedure_double]{$\cmaxcut$ kernel} procedure. Then the optimal value of the reduced instance, conditioned on not selecting any super vertex, satisfies
\[
\widetilde{\opt} := \max_{\substack{S \subseteq \widetilde{V} \setminus S_{\text{super}} \\ |S \cap \widetilde{V}_i| = k_i\ \forall i \in [c]}} \delta(S) \ge (1 - 4c\varepsilon) \cdot \opt.
\]
\end{theorem}

\begin{proof}
Let $S^*$ be an optimal solution to the original instance with $\delta(S^*) = \opt$. For each part $i \in [c]$, define $H_i := H_{k_i/\varepsilon}^{(i)}$ as the top $k_i / \varepsilon$ vertices in $V_i$ by weighted degree (as defined in the preliminaries).

We will transform $S^*$ into a solution $S_T$ such that $S_T \cap V_i \subseteq H_i$ for every $i \in [c]$ while losing only a small fraction of the cut value. At each step $t$, identify the smallest index $p \in [c]$ for which $S_t \cap V_p \nsubseteq H_p$, and apply the local exchange from \Cref{cor:local_exchange_gen} to swap a vertex $j \in (S_t\cap V_p) \setminus H_p$ with a vertex $i \in H_p \setminus (S_t\cap V_p)$, yielding a new set $S_{t+1}$ with
\[
\delta(S_{t+1}) \ge \left(1 - \frac{2}{k_i/\varepsilon - k_i}\right) \cdot \delta(S_t).
\]
For each $i \in [c]$, we perform at most $k_i$ such exchanges in $V_i$. Hence, the total cut value at the end satisfies:
\[
\delta(S_T) \ge \delta(S^*) \cdot \prod_{i=1}^c \left(1 - \frac{2}{k_i/\varepsilon - k_i}\right)^{k_i}.
\]
Using the inequality $(1 - x)^m \ge 1 - mx$, we get
\[
\delta(S_T) \ge \opt \cdot \left(1 - \sum_{i=1}^c \frac{2k_i}{k_i/\varepsilon - k_i}\right) \ge \opt \cdot \left(1 - \sum_{i=1}^c \frac{2\varepsilon}{1 - \varepsilon}\right) \ge \opt \cdot (1 - 4c\varepsilon),
\]
where the final bound assumes $\varepsilon \le 1/2$.
\end{proof}

\begin{lemma}
\label{lem:local_exchange_gen}
Let $S \subseteq V$ be a subset of size $k$ and let $H \subseteq V$ be a subset of size greater than $k$ such that $S \setminus H \ne \emptyset$ and every vertex in $H \setminus S$ has higher weighted degree than every vertex in $S \setminus H$. Then there exist $i \in H \setminus S$ and $j \in S \setminus H$ such that:
\[
\delta((S - j) + i) \ge \left(1 - \frac{2}{|H\backslash S|}\right) \cdot \delta(S).
\]
\end{lemma}
\begin{proof}
The proof is identical to that of \Cref{lem:local_exchange}. It proceeds by selecting $i$ to minimize $\delta(S, \{i\})$ over $H \setminus S$ and applying cut submodularity to bound the loss when replacing $j \in S \setminus H$.
\end{proof}
\begin{cor}
\label{cor:local_exchange_gen}
For any subset $S\subseteq V$ and an index $p\in [c]$ such that $|S\cap V_p|=k_p\leq h$ and $(S\cap V_p)\backslash H_{h}^{(p)}\neq \emptyset$, there exist vertices $i\in H_{h}^{(p)}\backslash (S\cap V_p)$ and $j\in (S\cap V_p)\backslash H_{h}^{(p)}$ such that 
\[
\delta((S - j) + i) \ge \left(1 - \frac{2}{h - k_p}\right) \cdot \delta(S).
\]
\end{cor}
\begin{proof}
    Using $H=(S\backslash V_p)\cup H_h^{(p)}$ in \Cref{lem:local_exchange_gen}, and the fact that $H\backslash S= H_{h}^{(p)}\backslash (S\cap V_p)$ finishes the proof. 
\end{proof}

\section{Single Constraint}

In this section, we describe our $(\alpha_{cc} - \varepsilon)$-approximation algorithm for $\maxcut_k$, for all values of $k$. Without loss of generality, we assume $k \leq n/2$ due to the symmetry of the cut function.

\subsection{Algorithm}
\label{subsec:Alg_single}

\textbf{Input:} Weighted graph $G = (V, E)$ and parameters $k \leq n/2$, $0 < \varepsilon \leq 1/2$. \\
\textbf{Output:} A set $S \subseteq V$ of size $|S| = k$.

\begin{enumerate}[leftmargin=*]
    \item \textbf{(Preprocessing Step)}  
    Let $\widetilde{G} = (\widetilde{V}, \widetilde{E})$ be the approximate kernel output by the \hyperref[subsec:kernel_max_cut_k]{$\maxcut_k$~kernel} with input $(G, k, \varepsilon)$. Note that $|\widetilde{V}| = O(k/\varepsilon)$.

    \item \textbf{(SDP and Conditioning)}  
    \begin{enumerate}
        \item Solve a $(3 + 4/\varepsilon^{120})$-round Lasserre SDP relaxation for the $\maxcut_k$ problem on the graph $\widetilde{G}$ (see \Cref{subsec:lassere_sdp_single}).
        \item Apply \Cref{thm:valuepreserving_independent_single} with $\alpha = \varepsilon^{60}$ and $\ell = 2$ to obtain a 2-level SDP solution that is $\varepsilon^{60}$-independent and has objective value at least $\widetilde{\opt} - \varepsilon^{60}$, where $\widetilde{\opt}$ is the optimum value of $\maxcut_k$ on $\widetilde{G}$ (conditioned on not selecting the super vertex $s$). By \Cref{lem:lowerbound_opt_single}~(\ref{lem:structural_2_single}), we know that $\widetilde{\opt} - \varepsilon^{60} \ge (1 - \varepsilon) \widetilde{\opt}$.
    \end{enumerate}

    \item \textbf{(Rounding)}  
    Apply the rounding algorithm of Raghavendra and Tan (Algorithm 5.3 in~\cite{raghavendra2012approximating}) to obtain a (random) set $\widehat{S}$. Let $\mathcal{E}$ denote the event that $|\widehat{S}| \in [k - \varepsilon^2 |\widetilde{V}|,\, k + \varepsilon^2 |\widetilde{V}|]$.

    \item \textbf{(Correction)}  
    If event $\mathcal{E}$ does not occur, return an arbitrary subset $S \subseteq \widetilde{V} \setminus \{s\}$ of size $k$. Otherwise, adjust $\widehat{S}$ by randomly adding or removing vertices to produce a set $S$ of size exactly $k$, and return $S$.
\end{enumerate}

\subsection{Lasserre SDP}
\label{subsec:lassere_sdp_single}
We now describe the SDP used in Step 2 of the \hyperref[subsec:Alg_single]{algorithm} above. Since a full overview of the Lasserre hierarchy is already provided in \Cref{sec:Lasserre-basics}, we only describe the relevant formulation.

After the preprocessing step, we solve the following level-$\ell$ Lasserre SDP relaxation for $\maxcut_k$ on the reduced graph $\widetilde{G}$, with an additional constraint ensuring the super vertex $s$ is not selected:

\begin{align}
 \label{SDP:maxcut_k_conditioned}   
\max\quad & \sumL_{\{i,j\} \in E} w_{i,j} \cdot \PP_{\mu_{\{i,j\}}}[X_{\{i,j\}} \in \{(-1,1), (1,-1)\}] \\
\text{s.t.}\quad 
& \sumL_{i \in \widetilde{V}} \PP_{\mu_{S \cup \{i\}}}(X_i = 1 \mid X_S = \alpha) = k && \forall S \subseteq V,\ |S| \le \ell - 1,\ \alpha \in \{0,1\}^S \nonumber \\
& \PP_{\mu_{S \cup \{s\}}}(X_s = 1 \mid X_S = \alpha) = 0 && \forall |S| \le \ell - 1,\ \alpha \in \{0,1\}^S \\
& \mu \textnormal{ is a level-} \ell \textnormal{ pseudo-distribution}. \nonumber
\end{align}

\subsection{Analysis}

\begin{proof}[Proof of \Cref{thm:Max-Cut-k}]
Using \Cref{lem:lowerbound_opt_single}~(\ref{lem:structural_1_single}), the optimal value on the kernelized instance $\widetilde{G}$ is at least $(1 - 4\varepsilon)\opt$. After solving the SDP, we obtain a value at least $(1 - \varepsilon)\widetilde{\opt} \ge (1 - 5\varepsilon)\opt$.

The expected size of $\widehat{S}$ is exactly $k$ since Algorithm 5.3 from \cite{raghavendra2012approximating} preserves the bias of each vertex. Using \Cref{thm:balance_variance}, the variance of the balance $|\widehat{S}|/|\widetilde{V}|$ is at most $O(\varepsilon^{60/12})=O(\varepsilon^5)$ (Assume that the constant hidden in the $O$-notation is $1$ for simplicity. One can absorb the constant into the $\varepsilon$ in general). By Chebyshev’s inequality, the event $\mathcal{E}$ occurs with probability at least $1 - \varepsilon$.

The expected cut value conditioned on $\mathcal{E}$ satisfies:
\begin{align}
\label{eqn:conditional_S_hat}
\EE[\delta(\widehat{S}) \mid \mathcal{E}] 
\geq \frac{\EE[\delta(\widehat{S})] - \varepsilon \cdot \opt}{1 - \varepsilon}
\geq \left( \frac{(1 - 5\varepsilon)\alpha_{cc} - \varepsilon}{1 - \varepsilon} \right)\cdot \opt
\geq (\alpha_{cc} - 9\varepsilon)\cdot \opt.
\end{align}

Let $S$ be the final corrected set. Then:
\begin{align}
\EE[\delta(S)] 
&\geq (1 - \varepsilon) \cdot \EE[\delta(S) \mid \mathcal{E}] \\
&\stackrel{(\cref{lem:randomcorrection_single})}{\geq} (1 - \varepsilon)^2 \cdot \EE[\delta(\widehat{S}) \mid \mathcal{E}] \\
&\stackrel{(\ref{eqn:conditional_S_hat})}{\geq} (1 - \varepsilon)^2 \cdot (\alpha_{cc} - 9\varepsilon) \cdot \opt.
\end{align}

Steps 1 and 4 take $O(|E|\log|E|)$ and $O(|V|)$ time, respectively. Steps 2 and 3, which involve solving the Lasserre SDP and rounding, dominate the runtime and require $O(|\widetilde{V}|)^{\poly(1/\varepsilon)}$ time.
\end{proof}

\begin{lemma}
\label{lem:randomcorrection_single}
Conditioned on the event $\mathcal{E}$ and a fixed $\widehat{S}$, let $S$ be the set obtained by randomly adding or deleting $|k - |\widehat{S}||$ vertices from $\widetilde{V} \setminus (\widehat{S} \cup \{s\})$ so that $|S| = k$. Then:
\[
\EE[\delta(S)] \geq (1 - \varepsilon) \cdot \delta(\widehat{S}).
\]
\end{lemma}

\begin{lemma}
\label{lem:lowerbound_opt_single}
Let $\widetilde{\opt}$ denote the optimum value of $\maxcut_k$ on $\widetilde{G}$, conditioned on not selecting the super vertex $s$.
\begin{enumerate}
    \item \label{lem:structural_1_single} $\widetilde{\opt} \geq (1 - 4\varepsilon) \cdot \opt$, where $\opt$ is the optimum value for $\maxcut_k$ on $G$.
    \item \label{lem:structural_2_single} $\widetilde{\opt}$ is at least an $\varepsilon$-fraction of the total edge weight in $\widetilde{E}$.
\end{enumerate}
\end{lemma}

\begin{proof}
Part~(\ref{lem:structural_1_single}) follows from \Cref{thm:maxcut_kernel_red_single}. For part~(\ref{lem:structural_2_single}), we show that a uniformly random subset of $V \setminus \{s\}$ of size $k$ cuts any edge with probability at least $\varepsilon$.

Let $n' := |\widetilde{V} \setminus \{s\}|$. Then $2k \le n' \le k/\varepsilon$. If edge $e$ is adjacent to $s$, it is cut with probability $k/n' \ge \varepsilon$. Otherwise, the cut probability is $2k(n' - k)/(n'(n' - 1)) \ge k/n' \ge \varepsilon$.
\end{proof}

\section{Constant Number of Constraints}

In this section, we present our $(\alpha_{cc} - \varepsilon)$-approximation algorithm for $\cmaxcut$, the Max-Cut problem with $c$ cardinality constraints. Our primary focus is on instances where the number of vertices to be selected from each part $V_i$ is relatively small, and for this reason, we assume that $k_i \leq n_i/2$. (Unlike the case in $\maxcut_k$, this is not without loss of generality.)

\smallskip
The key observation enabling this extension of~\cite{raghavendra2012approximating} to multiple constraints is that the notion of $\alpha$-independence can be defined locally within each block. Specifically, it suffices to ensure that the average mutual information between vertex pairs within each part is small:
\[
\mathbb{E}_{i, i' \in V_j}[I(X_i; X_{i'})] \leq \alpha \quad \text{for all } j \in [c].
\]
If this condition holds, then after rounding via Algorithm 5.3 of~\cite{raghavendra2012approximating}, the size of each intersection $|\widehat{S} \cap V_j|$ concentrates around its expectation. In particular, using \Cref{thm:balance_variance} with $W$ as the uniform distribution over $V_j$, the variance of $|\widehat{S} \cap V_j| / |V_j|$ is bounded by $O(\alpha^{1/12})$ for every $j \in [c]$. Therefore, for an appropriate choice of $\alpha$, we obtain
\[
|\widehat{S} \cap V_j| \in [k_j(1 - \varepsilon),\; k_j(1 + \varepsilon)]
\quad \text{simultaneously for all } j \in [c],
\]
with probability at least $1 - \varepsilon$.

\medskip

\begin{defn}[$\alpha$-block independence]
An SDP solution to an $\ell$-round Lasserre relaxation is \emph{$\alpha$-block independent} if
$\mathbb{E}_{i, i' \in V_j}[I_{\mu_{\{i,i'\}}}(X_i; X_{i'})] \leq \alpha$ 
hold for all $j \in [c]$.
\end{defn}

To find such solutions, we extend the conditioning technique of~\cite{raghavendra2012approximating}. The following procedure begins with an $(L + \ell)$-round SDP solution and returns an $\ell$-round $\alpha$-block independent solution for $L = O(c^2 / \alpha^2)$.

\paragraph{Conditioning Procedure:}
\begin{enumerate}
    \item For each $t \in [L]$:
    \begin{enumerate}
        \item Sample a block index $j_t \in [c]$ uniformly at random.
        \item Sample a vertex $i_t \in V_{j_t}$ uniformly at random.
        \item Sample $X_{i_t}$ from its marginal distribution under the current SDP solution (conditioned on previous outcomes), and condition on this value.
        \item If the resulting SDP solution is $\alpha$-block independent, terminate and return it.
    \end{enumerate}
\end{enumerate}

\begin{lemma}
\label{lem:sampling_correlation_reduction_block}
For any $L \in \mathbb{Z}_{\geq 2}$, there exists $t \leq L$ such that
\begin{align*}
\mathbb{E}_{j_1, \dots, j_t \in [c]} \mathbb{E}_{i_1 \in V_{j_1}, \dots, i_t \in V_{j_t}} \left[
\sum_{j, j' \in [c]} \mathbb{E}_{i \in V_j,\; i' \in V_{j'}} 
I(X_i; X_{i'} \mid X_{i_1}, \dots, X_{i_t})
\right] \leq \frac{c^2}{L}.
\end{align*}
\end{lemma}

\begin{proof}
Define the potential function:
\[
\phi_t := \mathbb{E}_{j_1, \dots, j_t \in [c]} \mathbb{E}_{i_1 \in V_{j_1}, \dots, i_t \in V_{j_t}} 
\left[ \mathbb{E}_{j \in [c]} \mathbb{E}_{i \in V_j} H(X_i \mid X_{i_1}, \dots, X_{i_t}) \right].
\]

Now, conditioned on fixed values of $j_1, \dots, j_t$ and $i_1, \dots, i_t$, the difference in potentials is:
\begin{align*}
\phi_t - \phi_{t+1} &= \mathbb{E}_{j \in [c]} \mathbb{E}_{i \in V_j} 
\left(H(X_i \mid X_{i_1}, \dots, X_{i_t}) - \mathbb{E}_{j_{t+1} \in [c]} \mathbb{E}_{i_{t+1} \in V_{j_{t+1}}}
H(X_i \mid X_{i_1}, \dots, X_{i_{t+1}}) \right) \\
&= \mathbb{E}_{j, j_{t+1} \in [c]} \mathbb{E}_{i \in V_j, i_{t+1} \in V_{j_{t+1}}}
I(X_i; X_{i_{t+1}} \mid X_{i_1}, \dots, X_{i_t}) \\
&= \frac{1}{c^2} \sum_{j, j' \in [c]} \mathbb{E}_{i \in V_j, i' \in V_{j'}} 
I(X_i; X_{i'} \mid X_{i_1}, \dots, X_{i_t}).
\end{align*}

Taking expectation over all random choices of $j_1, \dots, j_t$ and $i_1, \dots, i_t$ gives:
\begin{align}
\label{eqn:diff_potential_block}
\phi_t - \phi_{t+1} = \frac{1}{c^2} 
\mathbb{E}_{j_1, \dots, j_t} \mathbb{E}_{i_1, \dots, i_t}
\left[ \sum_{j, j' \in [c]} \mathbb{E}_{i \in V_j, i' \in V_{j'}} 
I(X_i; X_{i'} \mid X_{i_1}, \dots, X_{i_t}) \right].
\end{align}

Summing \eqref{eqn:diff_potential_block} over $t = 0$ to $L - 1$, and noting that entropy is always non-negative, we get:
\[
\sum_{t=0}^{L-1} 
\mathbb{E}_{j_1, \dots, j_t} \mathbb{E}_{i_1, \dots, i_t} 
\left[ \sum_{j, j' \in [c]} \mathbb{E}_{i \in V_j, i' \in V_{j'}} 
I(X_i; X_{i'} \mid X_{i_1}, \dots, X_{i_t}) \right]
\leq c^2 (\phi_0 - \phi_L) \leq c^2.
\]

Therefore, by averaging, there exists $t \leq L$ for which the expected blockwise mutual information is at most $c^2 / L$.
\end{proof}

\begin{cor}
\label{cor:avg_cross_block_to_block}
If
\[
\sum_{j,j' \in [c]} \mathbb{E}_{i \in V_j, i' \in V_{j'}} I(X_i; X_{i'} \mid X_{i_1}, \dots, X_{i_t}) \leq \alpha,
\]
then
\[
\mathbb{E}_{i, i' \in V_j} I(X_i; X_{i'} \mid X_{i_1}, \dots, X_{i_t}) \leq \alpha \quad \text{for all } j \in [c].
\]
\end{cor}

\begin{proof}
Each blockwise term is a subset of the global sum, and mutual information is non-negative.
\end{proof}

\newcommand{\val}{\operatorname{SDP}_{\ell}}

\begin{theorem}
\label{thm:valuepreserving_independent_double}
For every $\alpha > 0$ and integer $\ell > 0$, there exists an algorithm running in time $O(n^{\ell + \poly(c/\alpha)})$ that finds an $\alpha$-block independent solution to the $\ell$-round Lasserre SDP with value at least $\opt - \alpha$, where $\opt$ is the optimum value of the $(L + \ell)$-round SDP.
\end{theorem}
\begin{proof}
Set $L = \frac{4c^2}{\alpha^2}$. First, solve the $(L + \ell)$-round Lasserre SDP relaxation (as described in \Cref{subsec:Alg_double}) to obtain an initial solution.

Next, apply the conditioning procedure described above. That is, for each $t \in [L]$, sample a block index $j_t \in [c]$ uniformly at random, then sample a vertex $i_t \in V_{j_t}$ uniformly, sample $X_{i_t}$ from its marginal distribution (after the first $t-1$ fixings), and condition the SDP solution on that assignment. Continue this process until the resulting pseudo-distribution becomes $\alpha$-block independent.

We analyze this procedure by appealing to \Cref{lem:sampling_correlation_reduction_block}, which shows that:
\begin{align*}
\mathbb{E}_{j_1, \dots, j_t \sim [c]} \mathbb{E}_{i_1 \sim V_{j_1}, \dots, i_t \sim V_{j_t}} 
\left[\sum_{j,j' \in [c]} \mathbb{E}_{i \in V_j,\; i' \in V_{j'}} I(X_i; X_{i'} \mid X_{i_1}, \dots, X_{i_t}) \right] 
\leq \frac{c^2}{L} = \frac{\alpha^2}{4}.
\end{align*}
For some $t\leq L$.
By Markov's inequality, the probability that the total conditional mutual information (summed over all block pairs) exceeds $\alpha$ is at most:
\[
\Pr\left[\sum_{j,j'} \mathbb{E}_{i\in V_j, i'\in V_{j'}} I(X_i; X_{i'} \mid X_{i_1}, \dots, X_{i_t}) > \alpha \right] 
\leq \frac{\alpha^2 / 4}{\alpha} = \frac{\alpha}{4}.
\]

Thus, with probability at least $1 - \frac{\alpha}{4}$, the conditioned solution is $\alpha$-block independent. By \Cref{cor:avg_cross_block_to_block}, this also implies that:
\[
\mathbb{E}_{i,i' \in V_j} I(X_i; X_{i'} \mid X_{i_1}, \dots, X_{i_t}) \leq \alpha \quad \text{for all } j \in [c],
\]
so the solution satisfies the desired independence property within each block.

Now consider the effect of the conditioning procedure on the SDP objective value. Let $\val$ denote the value of the SDP after conditioning. Since conditioning preserves expectations, we have $\mathbb{E}[\val] = \opt$. To bound the probability that the value drops by more than $\alpha$, we apply Markov's inequality to the non-negative random variable $1 - \val$:
\[
\Pr[\val < \opt - \alpha] = \Pr[1 - \val > 1 - \opt + \alpha] \leq \frac{1 - \opt}{1 - \opt + \alpha} \leq \frac{1}{1 + \alpha},
\]
where the last inequality uses $\opt \leq 1$.

Separately, as shown earlier, the probability that the conditioned solution fails to be $\alpha$-block independent is at most $\alpha/4$. By a union bound, the total failure probability is
\[
\frac{\alpha}{4} + \frac{1}{1 + \alpha} < 1,
\]
for all $\alpha \le 1$. Hence, there exists a choice of conditioning—i.e., some $t \le L$ and assignment to $X_{i_1}, \dots, X_{i_t}$—such that the resulting SDP solution is $\alpha$-block independent and has objective value at least $\opt - \alpha$.

This outcome can be found by brute-force search over all subsets of up to $L = O(c^2/\alpha^2)$ variables and their possible assignments. The overall runtime is thus
\[
O(n^{\ell + L}) = O(n^{\ell + \poly(c/\alpha)}),
\]
as claimed.
\end{proof}

\begin{proof}[Proof of~\Cref{thm:Max-Cut-constant-constraints}]
    \Cref{lem:main_double} proves this theorem for Algorithm~\ref{subsec:Alg_double}. 
\end{proof}

\subsection{Algorithm}
\label{subsec:Alg_double}

\textbf{Input:} Weighted graph $G = (V, E)$ with vertex set partitioned as $V = \biguplus_{i=1}^c V_i$, parameters $k_i \le |V_i|/2$ for $i \in [c]$, and $0 < \varepsilon \le 1/2$.\\
\textbf{Output:} A set $S \subseteq V$ such that $|S \cap V_i| = k_i$ for all $i \in [c]$.

\begin{enumerate}
    \item \textbf{(Preprocessing Step)}  
    Let $\widetilde{G} = (\widetilde{V}, \widetilde{E})$ be the approximate kernel obtained via the \linebreak \hyperref[sec:kernal_procedure_double]{$\cmaxcut$ kernel} procedure with input $(G, (k_1,\dots,k_c), \varepsilon)$.

    \item \textbf{(SDP and Conditioning)}  
    \begin{enumerate}
        \item Solve a $(3 + 4c/\varepsilon^{120})$-round Lasserre SDP relaxation for $\cmaxcut$ on $\widetilde{G}$ (see \Cref{subsec:lassere_sdp_double}).
        \item Apply \Cref{thm:valuepreserving_independent_double} with $\alpha = \varepsilon^{60}$ and $\ell = 2$ to obtain a 2-level SDP solution that is $\varepsilon^{60}$-block independent and has value at least $\widetilde{\opt} - \varepsilon^{60}$. From \Cref{lem:lowerbound_opt_double}, we know $\widetilde{\opt} - \varepsilon^{60} \ge (1 - \varepsilon)\widetilde{\opt}$.
    \end{enumerate}

    \item \textbf{(Rounding)}  
    Apply Algorithm 5.3 from~\cite{raghavendra2012approximating} to obtain a (random) set $\widehat{S}$. Let $\mathcal{E}_i$ denote the event that $|\widehat{S} \cap V_i| \in [k_i - \varepsilon^2 |V_i|,\; k_i + \varepsilon^2 |V_i|]$ for each $i \in [c]$, and define $\mathcal{E} := \bigcap_{i=1}^c \mathcal{E}_i$.

    \item \textbf{(Correction)}  
    If $\mathcal{E}$ does not occur, return an arbitrary feasible set. Otherwise, for each part $V_i$, randomly add or remove vertices to ensure $|\widehat{S} \cap V_i| = k_i$, and return the resulting set $S$.
\end{enumerate}
\begin{lemma}
\label{lem:main_double}
    The expected value of the cut returned by Algorithm~\ref{subsec:Alg_double} is at least $\big(\alpha_{cc}-O(\varepsilon)\big)\opt$. The running time of the algorithm is $O\big(\min\{k/\varepsilon,n\}^{\poly(c/\varepsilon)}+\poly(n)\big)$ where $k=\sumL_{i=1}^c k_i$.    
\end{lemma}
\subsection{SDP Relaxation}
\label{subsec:lassere_sdp_double}

We solve the following level-$\ell$ Lasserre SDP relaxation for $\cmaxcut$ on the reduced graph $\widetilde{G}$:

\begin{align}
\label{SDP:maxcut_partition_conditioned}
\max\quad & \sumL_{\{i,j\} \in E} w_{i,j} \cdot \PP_{\mu_{\{i,j\}}}[X_{\{i,j\}} \in \{(-1,1), (1,-1)\}] \\
\text{s.t.}\quad
& \sumL_{i \in \widetilde{V}_j} \PP_{\mu_{S \cup \{i\}}}(X_i = 1 \mid X_S = \alpha) = k_j && \forall j \in [c],\; |S| \le \ell - 1,\; \alpha \in \{0,1\}^S \nonumber \\
& \PP_{\mu_{S \cup \{s_j\}}}(X_{s_j} = 1 \mid X_S = \alpha) = 0 && \forall j \in [c],\; |S| \le \ell - 1,\; \alpha \in \{0,1\}^S \nonumber \\
& \mu \textnormal{ is a level-}\ell \textnormal{ pseudo-distribution}. \nonumber
\end{align}

\subsection{Analysis}
\begin{proof}[Proof of \Cref{lem:main_double}]
By \Cref{lem:lowerbound_opt_double}, we have $\widetilde{\opt} \ge (1 - 4c\varepsilon)\opt$, and after solving the SDP and conditioning, the objective remains at least $(1 - \varepsilon)\widetilde{\opt} \ge (1 - (4c + 1)\varepsilon)\opt$.

Since the SDP solution is $\varepsilon^{60}$-block independent, using \Cref{thm:balance_variance} with $W$ as the uniform distribution over each $V_i$, the variance of $|\widehat{S} \cap V_i| / |V_i|$ is $O(\varepsilon^5)$. By Chebyshev's inequality, the event $\mathcal{E}_i$ occurs with probability at least $1 - \varepsilon$, so the joint event $\mathcal{E} = \bigcap_{i=1}^c \mathcal{E}_i$ occurs with probability at least $1 - c\varepsilon$.

The expected value of the cut after rounding is at least $\alpha_{cc} \cdot (1 - (4c+1)\varepsilon) \cdot \opt$. Conditioning on $\mathcal{E}$, we have:
\begin{align}
\label{eqn:conditional_S_hat_double}
\EE[\delta(\widehat{S}) \mid \mathcal{E}]
\ge \frac{\EE[\delta(\widehat{S})] - c\varepsilon \cdot \opt}{1 - c\varepsilon}
\ge \left( \frac{\alpha_{cc}(1 - (4c + 1)\varepsilon) - c\varepsilon}{1 - c\varepsilon} \right) \cdot \opt
\ge (\alpha_{cc} - O(c\varepsilon)) \cdot \opt.
\end{align}

Let $S$ be the corrected set after adjusting $\widehat{S}$ to satisfy cardinality constraints exactly. Using the same argument as in \Cref{lem:randomcorrection_single} applied sequentially across the $c$ parts, we get:
\begin{align}
\EE[\delta(S)]
\ge (1 - c\varepsilon) \cdot (1 - \varepsilon)^c \cdot \EE[\delta(\widehat{S}) \mid \mathcal{E}]
\ge (\alpha_{cc} - O(c\varepsilon)) \cdot \opt.
\end{align}

The total running time is dominated by solving the SDP and brute-force conditioning, which takes $O(n^{\ell + \poly(c/\varepsilon)})$. Preprocessing and postprocessing steps take $\poly(n)$ time.
\end{proof}

\begin{lemma}
\label{lem:lowerbound_opt_double}
Let $\widetilde{\opt}$ be the optimum value of $\cmaxcut$ on $\widetilde{G}$ (conditioned on not picking any $s_i$).
\begin{enumerate}
    \item \label{lem:structural_1_double} $\widetilde{\opt} \ge (1 - 4c\varepsilon)\cdot \opt$, where $\opt$ is the optimum value of the $\cmaxcut$ instance on $G$.
    \item \label{lem:structural_2_double} $\widetilde{\opt} \ge \varepsilon$ fraction of the total edge weight in $\widetilde{E}$.
\end{enumerate}
\end{lemma}
\begin{proof}
Part~(\ref{lem:structural_1_double}) follows from \Cref{thm:maxcut_kernel_red_double}. For~(\ref{lem:structural_2_double}), consider sampling $S = \bigcup_{i=1}^c S_i$, where each $S_i$ is a uniformly random subset of size $k_i$ from $\widetilde{V}_i \setminus \{s_i\}$. Let $n_i' = |\widetilde{V}_i \setminus \{s_i\}|$. Since $2k_i \le n_i' \le k_i/\varepsilon$, we have $k_i / n_i' \ge \varepsilon$.

\begin{enumerate}
    \item If an edge is adjacent to $s_i$, the cut probability is at least $\varepsilon$.
    \item If both endpoints are in the same $V_i$, \Cref{lem:lowerbound_opt_single}~(\ref{lem:structural_2_single}) gives cut probability $\ge \varepsilon$.
    \item If endpoints lie in $V_i$ and $V_j$ ($i \ne j$), the probability that exactly one endpoint lies in $S$ is
  \[
  \left(\frac{k_i}{n_i'}\right)\left(1 - \frac{k_j}{n_j'}\right) + \left(\frac{k_j}{n_j'}\right)\left(1 - \frac{k_i}{n_i'}\right) \ge \varepsilon/2+\varepsilon/2=\varepsilon.
  \]

\end{enumerate}
So the expected cut value of $S$ is at least an $\varepsilon$ fraction of total edge weight in $\widetilde{E}$.
\end{proof}

\begin{proof}[Proof of~\Cref{thm:Max-Cut-constant-constraints}]
\Cref{lem:main_double} proves this theorem for Algorithm~\ref{subsec:Alg_double}.
\end{proof}

\section{Arbitrary Number of Constraints}
We consider the general case of $\cmaxcut$ with an arbitrary number of constraints, potentially $c = \omega(1)$. First, we present a $0.5$-approximation for the more general problem of Max-Cut under an arbitrary matroid constraint. 
Next, we establish an NP-hardness result for determining whether the optimal solution in a given instance of $\cmaxcut$ with an arbitrary number of constraints equals the total number of edges in the graph.  

\subsection{Approximation Algorithm}
\begin{proof}[Proof of~\Cref{thm:Max-Cut-matroid}]
    Consider the following linear program:
    \begin{align}
      \label{LP:maxcut-base}  \max:& \sumL_{e\in E}w_e y_e \\
        y_{\{u,v\}}& \leq x_u+x_v \\
        y_{\{u,v\}}& \leq 2-(x_u+x_v) \\
        x&\in \cB
    \end{align}
    where $\cB$ is the base polytope of the matroid $\cM$. We can see that for any given $x$, the optimal choice for $y_{\{u,v\}}$ is $\min\{x_u+x_v, 2-x_u-x_v\}$. When $x$ is integral, this function also coincides with the indicator whether edge $\{u,v\}$ has been cut. Since the matroid polytope is separable, we can solve the LP \Cref{LP:maxcut-base} efficiently. 
    
    Now consider the following non-concave quadratic program:
            \begin{align}
      \label{QP:maxcut-base}  \max:& \sumL_{\{u,v\}\in E}w_e \left(x_u+x_v-2x_ux_v\right) \\
        x&\in \cB
    \end{align}
Observe that the function $x_u+x_v-2x_ux_v$ also coincides with the cut indicator function of edge $\{u,v\}$ when $x$ is integral. Even though we cannot solve \Cref{QP:maxcut-base} efficiently, we can show that it has no integrality gap and infact that we can round any fractional solution $\widehat{x}\in \cB$ to a solution $x\in \cB$ that is integral and with value at least that of $\widehat{x}$. The two crucial properties we need that are easy to see are:
\begin{enumerate}
    \item The function $\sumL_{\{u,v\}\in E}w_e \left(x_u+x_v-2x_ux_v\right)$ is convex in any direction $e_u-e_v$ for $u\neq v \in V$. Here $e_u\in \{0,1\}^V$ is the indicator vector for vertex $u$.  
    \item The polytope $\cB$ is the facet of a matroid and hence solvable and integral.
\end{enumerate}
Given these properties, any fraction solution can be pipage rounded (see \cite{AS04} and especially section 3.2 of \cite{CCPV11}) to an integral solution with value at least that of the fractional solution. The final observation is that for any $x\in [0,1]^V$, we have 
\begin{align}
   (x_u+x_v-2x_ux_v)\leq  \min\{x_u+x_v,2-x_u-x_v\}\leq 2(x_u+x_v-2x_ux_v).
\end{align}
from \Cref{lem:inequality}. This implies that the integrality gap of \Cref{LP:maxcut-base} is at most $0.5$ and in fact provides a way to find a rounding with value at least $0.5$ times the LP value. Solve the LP and pipage round the solution using the quadratic objective. 
\end{proof}
Note that the proof idea for \Cref{thm:Max-Cut-matroid} is essentially the same as in \cite{AS04} used for the Hypergraph Max $k$-cut with given sizes of parts problem and the pipage rounding for matroids from \cite{CCPV11}. 

\subsection{Hardness Result}

\begin{proof}[Proof of Theorem~\ref{thm:Max-Cut-hardness}]
We show a reduction from the 3D matching problem. An instance of the 3D matching problem is a tripartite graph with parts $X,Y,Z$. The edges are triples $(x,y,z)\in X\times Y\times Z$. The problem is to decide if there is a subset of the edges such that every vertex is included in exactly one edge. 

The reduction is as follows: For every edge $e=(x,y,z)$, consider the star graph with four vertices with the center labeled $e$ and the leaves labeled $(e,x), (e,y), (e,z)$ respectively. The overall graph $G$ is simply the union of all these stars. The partition matroid consists of parts $\cP_x$ for every vertex $x\in X$ that contains the vertices  $(e',x')$ such that $x'=x$. We have parts similarly for elements in $Y$ and $Z$. The capacity of every part is exactly $1$. 

(Completeness) If there is a collection of edges $e_i, i\in M$ that every element of $X,Y,Z$ is in exactly one edge, then consider the solution $S=\bigcup_{i\in M}\{(e_i,e_i.x),(e_i,e_i.y),e_i.z\}\cup_{i\notin M}\{e_i\} $ where $e_i.x:= e_i \cap X$. It is easy to see that $\delta(S)=1$ and $|S\cap \cP_u|=1$ for every $u\in X\cup Y\cup Z$. 

(Soundness) Since every star graph is bipartite, any solution such that $\delta(S)=1$ should have the center on one side and the leaves on the other side. This implies that every solution such that $\delta(S)=1$ is of the form $S=\bigcup_{i\in M}\{(e_i,e_i.x),(e_i,e_i.y),e_i.z\}\cup_{i\notin M}\{e_i\}$ for some $M\subset E$ where $E$ is the collection of triples from the 3D matching instance. The partition matroid constraint that $|S\cap P_u|=1$ for $u\in X\cap Y\cap Z$ exactly translates to $M$ being a perfect 3D matching. 
Since it is NP-Hard to decide if there is a perfect 3D matching, it is NP-hard to decide if there is a cut $S\subseteq V$ such that $\delta(S)=1$ and $|S\cap \cP_i|=k_i,\, i\in [c]$ when $c=\omega(1)$.
\end{proof}

\begin{remark}\label{rem:decision-poly-solvable}
The above theorem shows that, in general, deciding whether the Max-Cut value equals the total number of edges in the graph is solvable in polynomial time when the number of constraints is constant. Moreover, the decision problem becomes solvable in quasi-polynomial time when the number of constraints is $\poly(\log n)$.
\end{remark}


\bibliographystyle{alpha}
\bibliography{bibdb}

\newcommand{\etalchar}[1]{$^{#1}$}
\begin{thebibliography}{KKMO07}

\bibitem[ABG16]{austrin2016better}
Per Austrin, Siavosh Benabbas, and Konstantinos Georgiou.
\newblock Better balance by being biased: A 0.8776-approximation for max
  bisection.
\newblock {\em ACM Transactions on Algorithms (TALG)}, 13(1):1--27, 2016.

\bibitem[AS04]{AS04}
Alexander~A. Ageev and Maxim Sviridenko.
\newblock Pipage rounding: A new method of constructing algorithms with proven
  performance guarantee.
\newblock {\em Journal of Combinatorial Optimization}, 8:307--328, 2004.

\bibitem[AS19]{austrin2019global}
Per Austrin and Aleksa Stankovi{\'c}.
\newblock Global cardinality constraints make approximating some max-2-csps
  harder.
\newblock {\em Approximation, Randomization, and Combinatorial Optimization.
  Algorithms and Techniques}, 145:24:1--24:17, 2019.

\bibitem[BFNS14]{buchbinder2014submodular}
Niv Buchbinder, Moran Feldman, Joseph Naor, and Roy Schwartz.
\newblock Submodular maximization with cardinality constraints.
\newblock In {\em Proceedings of the twenty-fifth annual ACM-SIAM symposium on
  Discrete algorithms}, pages 1433--1452. SIAM, 2014.

\bibitem[BRS11]{barak2011rounding}
Boaz Barak, Prasad Raghavendra, and David Steurer.
\newblock Rounding semidefinite programming hierarchies via global correlation.
\newblock In {\em Annual Symposium on Foundations of Computer Science}, pages
  472--481, 2011.

\bibitem[Cai08]{cai2008parameterized}
Leizhen Cai.
\newblock Parameterized complexity of cardinality constrained optimization
  problems.
\newblock {\em The Computer Journal}, 51(1):102--121, 2008.

\bibitem[CBO12]{candogan2012optimal}
Ozan Candogan, Kostas Bimpikis, and Asuman Ozdaglar.
\newblock Optimal pricing in networks with externalities.
\newblock {\em Operations Research}, 60(4):883--905, 2012.

\bibitem[CCPV11]{CCPV11}
Gruia Calinescu, Chandra Chekuri, Martin P\'{a}l, and Jan Vondr\'{a}k.
\newblock Maximizing a monotone submodular function subject to a matroid
  constraint.
\newblock {\em SIAM Journal on Computing}, 40(6):1740--1766, 2011.

\bibitem[FJ97]{frieze1997improved}
Alan Frieze and Mark Jerrum.
\newblock Improved approximation algorithms for {MAX $k$-CUT and MAX
  BISECTION}.
\newblock {\em Algorithmica}, 18(1):67--81, 1997.

\bibitem[FKL01]{feige2001note}
Uriel Feige, Marek Karpinski, and Michael Langberg.
\newblock A note on approximating max-bisection on regular graphs.
\newblock {\em Information Processing Letters}, 79(4):181--188, 2001.

\bibitem[FKP19]{FKP19-semialgebriac}
Noah Fleming, Pravesh Kothari, and Toniann Pitassi.
\newblock Semialgebraic proofs and efficient algorithm design.
\newblock {\em Foundations and Trends® in Theoretical Computer Science},
  14(1-2):1--221, 2019.

\bibitem[FL01]{feige2001approximation}
Uriel Feige and Michael Langberg.
\newblock Approximation algorithms for maximization problems arising in graph
  partitioning.
\newblock {\em Journal of Algorithms}, 41(2):174--211, 2001.

\bibitem[FL06]{feige2006rpr2}
Uriel Feige and Michael Langberg.
\newblock The $\mathrm{RPR}^2$ rounding technique for semidefinite programs.
\newblock {\em Journal of Algorithms}, 60(1):1--23, 2006.

\bibitem[FS14]{fotakis2014efficiency}
Dimitris Fotakis and Paris Siminelakis.
\newblock On the efficiency of influence-and-exploit strategies for revenue
  maximization under positive externalities.
\newblock {\em Theoretical Computer Science}, 539:68--86, 2014.

\bibitem[GMR{\etalchar{+}}11]{GMRSZ11}
Venkatesan Guruswami, Yury Makarychev, Prasad Raghavendra, David Steurer, and
  Yuan Zhou.
\newblock Finding almost-perfect graph bisections.
\newblock In {\em In Proceedings of Innovations in Computer Science}, pages
  321--337, 2011.

\bibitem[GW95]{goemans1995improved}
Michel~X Goemans and David~P Williamson.
\newblock Improved approximation algorithms for maximum cut and satisfiability
  problems using semidefinite programming.
\newblock {\em Journal of the ACM (JACM)}, 42(6):1115--1145, 1995.

\bibitem[HZ02]{HZ02}
Eran Halperin and Uri Zwick.
\newblock A unified framework for obtaining improved approximation algorithms
  for maximum graph bisection problems.
\newblock {\em Random Structures and Algorithms}, 20, 05 2002.

\bibitem[KKMO07]{khot2007optimal}
Subhash Khot, Guy Kindler, Elchanan Mossel, and Ryan O’Donnell.
\newblock Optimal inapproximability results for {MAX-CUT} and other
  $2$-variable {CSPs}?
\newblock {\em SIAM Journal on Computing}, 37(1):319--357, 2007.

\bibitem[Lau03]{Laurent-SA-LH-2003}
Monique Laurent.
\newblock A comparison of the sherali-adams, lovász-schrijver, and lasserre
  relaxations for 0-1 programming.
\newblock {\em Mathematics of Operations Research}, 28(3):470--496, 2003.

\bibitem[LLC22]{laclau2022survey}
Charlotte Laclau, Christine Largeron, and Manvi Choudhary.
\newblock A survey on fairness for machine learning on graphs.
\newblock {\em arXiv preprint arXiv:2205.05396}, 2022.

\bibitem[LMNS10]{LMNS09}
Jon Lee, Vahab~S. Mirrokni, Viswanath Nagarajan, and Maxim Sviridenko.
\newblock Maximizing nonmonotone submodular functions under matroid or knapsack
  constraints.
\newblock {\em SIAM Journal on Discrete Mathematics}, 23(4):2053--2078, 2010.

\bibitem[Rot13]{rothvoss2013lasserre}
Thomas Rothvo{\ss}.
\newblock The lasserre hierarchy in approximation algorithms.
\newblock {\em Lecture Notes for the MAPSP}, pages 1--25, 2013.

\bibitem[RT12]{raghavendra2012approximating}
Prasad Raghavendra and Ning Tan.
\newblock Approximating {CSPs} with global cardinality constraints using {SDP}
  hierarchies.
\newblock In {\em Proceedings of the Symposium on Discrete Algorithms}, pages
  373--387. SIAM, 2012.

\bibitem[SZ18]{saurabh2018k}
Saket Saurabh and Meirav Zehavi.
\newblock ($k, n-k$)-max-cut: An $\mathcal{O}^{*}(2^{p})$-time algorithm and a
  polynomial kernel.
\newblock {\em Algorithmica}, 80:3844--3860, 2018.

\bibitem[Ye01]{Ye01}
Yinyu Ye.
\newblock A $.699$-approximation algorithm for max-bisection.
\newblock {\em Mathematical Programming}, 90:101--111, 2001.

\end{thebibliography}

\appendix

\section{Basics of the Lasserre SDP Hierarchy}
\label{sec:Lasserre-basics}
For detailed expositions on the sum-of-squares hierarchy, we refer the reader to the excellent surveys by Laurent~\cite{Laurent-SA-LH-2003}, Rothvoß~\cite{rothvoss2013lasserre}, and Fleming et al.~\cite{FKP19-semialgebriac}. This section briefly summarizes key ideas drawn from these sources.

\medskip
Given a binary optimization problem with a linear relaxation defined by a matrix $\bm{A} \in \RR^{m \times n}$ and right-hand side $\bm{b} \in \RR^m$, consider the feasible region $K = \{x \in \RR^V : \bm{A}x \ge \bm{b}\}$. 

We ask: how can we systematically strengthen this relaxation to better approximate the convex hull of integral solutions, $\textnormal{conv}(K \cap \{-1,1\}^V)$?\footnote{We work over $\{-1,1\}^V$ rather than $\{0,1\}^V$ since this is the natural domain for problems like Max-Cut, where signs encode partition membership.}

Points in this convex hull can be interpreted as distributions over the hypercube $\{-1,1\}^V$. The level-$\ell$ Lasserre SDP yields a \textbf{pseudo-distribution} $\mu = \{\mu_S\}_{|S| \leq \ell}$, where each $\mu_S : \{-1,1\}^S \to [0,1]$ is a distribution over partial assignments to the subset $S \subseteq V$. However, there need not exist a global distribution whose marginals agree with these $\mu_S$. Despite this, the pseudo-distribution satisfies the following key properties:

\begin{enumerate}
    \item \textbf{Marginal Consistency:} The pseudo-distributions are consistent across overlapping subsets. That is, for any subsets $S, T \subseteq V$ with $|S|, |T| \leq \ell$ and any assignment $a \in \{-1,1\}^{S \cap T}$, we have:
    \begin{align}
        \mu_{S \cap T}(a) = \sumL_{b \in \{-1,1\}^{S \setminus T}} \mu_S(a \circ b) = \sumL_{c \in \{-1,1\}^{T \setminus S}} \mu_T(a \circ c),
    \end{align}
    where $a \circ b$ denotes the extension of $a$ to $S$ using $b$, and similarly for $a \circ c$ on $T$.

    \item \textbf{Conditioning:} The SDP solution supports conditioning on the value of any variable $i \in V$. Given a level-$\ell$ pseudo-distribution and variable $i$, there exist level-$(\ell-1)$ pseudo-distributions $\mu^{(+)}$, $\mu^{(-)}$ and a weight $\lambda \in [0,1]$ such that for all $S \subseteq V$ with $|S| \leq \ell - 1$ and $\alpha \in \{-1,1\}^S$,
    \begin{align}
        \mu_S(\alpha) = \lambda \cdot \mu_S^{(+)}(\alpha) + (1 - \lambda) \cdot \mu_S^{(-)}(\alpha),
    \end{align}
    where $\mu_S^{(+)}(\alpha)$ is nonzero only if $\alpha(i) = +1$ and $\mu_S^{(-)}(\alpha)$ is nonzero only if $\alpha(i) = -1$.
\end{enumerate}

While these properties are also satisfied by weaker hierarchies like Sherali–Adams, the Lasserre hierarchy is uniquely characterized by an additional sum-of-squares condition: for every polynomial $q(x)$ of degree at most $\ell$, the pseudo-expectation of its square is non-negative:
\begin{align}
\label{eqn:sum-of-squares-test}
    \EE_{\mu}[q(x)^2] \geq 0.
\end{align}
Assuming polynomials are multilinear (since we evaluate over the hypercube), any such polynomial $p(x)$ can be written as $p(x) = \sumL_S c_S y_S(x)$, where $y_S(x) := \prod_{i \in S} x_i$ and $|S| \leq \ell$. Then the pseudo-expectation becomes:
\[
\EE_{\mu}[p(x)] = \sumL_S c_S \sumL_{\alpha \in \{-1,1\}^S} \mu_S(\alpha) y_S(\alpha).
\]

Moreover, to incorporate the linear constraints $\bm{A}x \geq \bm{b}$, the Lasserre relaxation requires that:
\begin{align}
\label{eqn:constraints-test}
    \EE_{\mu}\left[q(x)^2 \cdot \left(\sumL_{j \in V} A_{i,j} x_j - b_i\right)\right] \geq 0 \quad \forall i \in [m], \text{ for all } q(x) \text{ of degree } \leq \ell - 1.
\end{align}

A level-$\ell$ pseudo-distribution satisfying \eqref{eqn:sum-of-squares-test} and \eqref{eqn:constraints-test} can be found by solving a semidefinite program, as described below.

\subsection{SDP Formulation}

\begin{defn}[Lasserre Hierarchy~\cite{rothvoss2013lasserre}]
Let $K = \{x \in \RR^V : \bm{A}x \geq \bm{b}\}$. The level-$t$ Lasserre relaxation, denoted $LAS_t(K)$, consists of vectors $y \in \RR^{2^{V}}$ satisfying:
\begin{align*}
    M_t(y) &= (y_{I \cup J})_{|I|, |J| \leq t} \succeq \bm{0}, \\
    M_t^{\ell}(y) &= \left( \sumL_{i=1}^n A_{\ell, i} y_{I \cup J \cup \{i\}} - b_\ell y_{I \cup J} \right)_{|I|, |J| \leq t-1} \succeq \bm{0} \quad \forall \ell \in [m], \\
    y_\emptyset &= 1.
\end{align*}
Here, $M_t(y)$ is the \emph{moment matrix}, and $M_t^{\ell}(y)$ is the \emph{moment matrix of slacks}. The projection onto the original variables is denoted by $LAS_t^{\textnormal{proj}}(K) := \{(y_{\{1\}}, \dots, y_{\{n\}}) : y \in LAS_t(K)\}$.

This is a valid relaxation: for any integral solution $x \in K \cap \{0,1\}^n$, the assignment $y_S = \prod_{i \in S} x_i$ for all $S \subseteq [n]$ yields a feasible point in $LAS_t(K)$.
\end{defn}

Each variable $y_S$ represents the pseudo-moment corresponding to all variables in $S$ being assigned $+1$. From these, one can recover the pseudo-distribution via Möbius inversion:
\begin{align}
    \mu_S(\mathbbm{1}_{S, S'}) = \sumL_{S \setminus S' \subseteq T \subseteq S} (-1)^{|T \cap S'|} y_T \quad \forall S' \subseteq S,
\end{align}
where $\mathbbm{1}_{S, S'} \in \{-1,1\}^S$ denotes the partial assignment that sets variables in $S'$ to $-1$ and the remaining in $S \setminus S'$ to $+1$.

\section{Omitted proofs}
\label{sec:omitted}
\begin{lemma}
    \label{lem:inequality}
    For any $x,y\in [0,1]$, $(x+y-2xy)\leq \min\{x+y,2-x-y\}\leq 2(x+y-2xy)$.
\end{lemma}
\begin{proof}
We consider the following cases,
    \begin{enumerate}
        \item When $x+y\leq 1$, the required inequality to prove is 
        \begin{align*}
            x+y-2xy \leq x+y \leq 2(x+y-2xy).
        \end{align*}
        The left most inequality is equivalent to $0\leq xy$ which is trivially true. The right most inequality is equivalent to $4xy\leq x+y$ which is true because $4xy \leq (x+y)^2 \leq x+y$. 
        \item When $x+y\leq 1$, substituting $x'=1-x,y'=1-y$, the required inequality to prove is 
        \begin{align*}
          x'+y'-2x'y' \leq  x'+y' \leq 2(x'+y'-2x'y').
        \end{align*}
        The conditions on $x',y'$ are that they are in the range $[0,1]$ and that $x'+y' \leq 1$. This is exactly the case above. 
    \end{enumerate}
\end{proof}

\begin{proof}[Proof of \Cref{lem:randomcorrection_single}]
Suppose $|\widehat{S}|< k$, then $S$ is obtained by adding a uniformly random set of $k-|\widehat{S}|$ vertices from $\widetilde{V}\setminus (\widehat{S}\cup \{s\})$ to $\widehat{S}$. The value $|\widetilde{V}|$ is equal to  $k/\varepsilon+1$ if $k/\varepsilon+1\leq n$ and $n$ otherwise. In the first case, the probability of an element in $\widetilde{V}\setminus \widehat{S}$ to be added to $\widehat{S}$ is at most $(k-|\widehat{S}|)/(k/\varepsilon -|\widehat{S}|)\leq \varepsilon$. In the second case, it is ${(k-|\widehat{S}|)}/{(n -|\widehat{S}|)}$. For it to be at most $\varepsilon$, it suffices to have $|\widehat{S}| \geq {(k-\varepsilon n)}/{(1-\varepsilon)}$,
which is true because in fact, we can show that
\begin{align*}
    \frac{k-\varepsilon n}{1-\varepsilon} \leq k-\varepsilon^2 n \leq |\widehat{S}|.
\end{align*}
The leftmost inequality is equivalent to $k\leq n(1-\varepsilon +\varepsilon^2)$ which is true because $k\leq n/2 \leq n(1-\varepsilon+\varepsilon^2)$. Using \Cref{lem:sampling_submodular}, we can imply that $\EE[\delta(S)]\geq (1-\varepsilon)\cdot \delta(\widehat{S}) $.

If $|\widehat{S}|> k$, then $S$ is obtained by removing a uniformly random subset of size $|\widehat{S}|-k$. The probability of an element in $\widehat{S}$ to be removed is $(|\widehat{S}|-k)/|\widehat{S}| \leq \varepsilon$. Hence, it suffices to have $|\widehat{S}|\leq \frac{k}{1-\varepsilon}$,
which is true because in fact, $|\widehat{S}|\leq k+\varepsilon^2 |\widetilde{V}|\leq k/(1-\varepsilon)$. The rightmost inequality is equivalent to $|\widetilde{V}|\leq k/(\varepsilon(1-\varepsilon))$. This is true because $|\widetilde{V}|\leq k/\varepsilon +1 \leq  k/(\varepsilon(1-\varepsilon))$. Using \Cref{lem:sampling_submodular} on $\widetilde{V}\setminus \widehat{S}$, we get $\EE[\delta(S)]\geq (1-\varepsilon)\cdot \delta(\widehat{S})$.
\end{proof}

\begin{lemma}[Lemma 2.2 in \cite{buchbinder2014submodular}]
\label{lem:sampling_submodular}
    For any set $S\subseteq V$, if $R\subseteq V\setminus S$ is a random set such that each element in $V\setminus S$ is included in $R$ (not necessarily independently) with probability at most $p$, then 
    \begin{align}
        \EE\left[\delta(S\cup R)\right] \geq (1-p)\cdot\delta(S).
    \end{align}
    This fact is generally true for any non-negative submodular function. 
\end{lemma}

\end{document}